\documentclass[a4paper,UKenglish]{lipics-v2019}
\usepackage{microtype}%if unwanted, comment out or use option "draft"
\usepackage{times,psfrag,epsf,epsfig,graphics,graphicx}
\usepackage{algorithm}
\usepackage{algorithmic}

\newcommand{\OPT}{\mbox{\sc OPT}}
\newcommand{\APP}{\mbox{\sc APP}}
\title{On Computing a Center Persistence Diagram}

\titlerunning{Computing a Center Persistence Diagram}

\author{Yuya Higashikawa}
{School of Social Information Science, University of Hyogo, Kobe, Japan}{higashikawa@sis.u-hyogo.ac.jp}{}{}

\author{Naoki Katoh}
{School of Social Information Science, University of Hyogo, Kobe, Japan}{naoki.katoh@gmail.com}{}{}

\author{Guohui Lin}{Department of Computing Science, University of Alberta, Edmonton, Alberta, Canada}{guohui@ualberta.ca}{}{}

\author{Eiji Miyano}
{Department of Artificial Intelligence, Kyushu Institute of Technology, Iizuka, Japan}{miyano@ces.kyutech.ac.jp}{}{}

\author{Suguru Tamaki}
{School of Social Information Science, University of Hyogo, Kobe, Japan}{tamak@sis.u-hyogo.ac.jp}{}{}

\author{Junichi Teruyama}
{School of Social Information Science, University of Hyogo, Kobe, Japan}{junichi.teruyama@sis.u-hyogo.ac.jp}{}{}

\author{Binhai Zhu}{Gianforte School of Computing, Montana State University, Bozeman, MT 59717, USA}{bhz@montana.edu}{}{}

\authorrunning{Higashikawa, et al.}

\Copyright{Yuya Higashikawa, Naoki Katoh, Guohui Lin, Eiji Miyano, Suguru Tamaki, Junichi Teruyama and Binhai Zhu}

\ccsdesc{Theory of computation - Design and analysis of algorithms}

\keywords{Persistence diagrams, Bottleneck distance, Center persistence diagram, NP-hardness, Approximation algorithms}

\category{}

\relatedversion{}

\supplement{}

\funding{}

\acknowledgements{This work was done when GL, EM and BZ visited University of Hyogo.}

\EventEditors{Artur Czumaj}
\EventNoEds{1}
\EventLongTitle{47th International Symposium on Automata, Languages and Programming (ICALP 2020)}
\EventShortTitle{ICALP 2020}
\EventAcronym{ICALP}
\EventYear{2020}
\EventDate{July 8--12, 2020}
\EventLocation{Beijing, China}
\EventLogo{}
\SeriesVolume{199}
\ArticleNo{1} % ¡°New number¡± (=<article-no>) goes here!
\nolinenumbers %uncomment to disable line numbering
\hideLIPIcs  %uncomment to remove references to LIPIcs series (logo, DOI, ...), e.g. when preparing a pre-final version to be uploaded to arXiv or another public repository
\begin{document}
\maketitle

\begin{abstract}
Throughout this paper, a persistence diagram ${\cal P}$ is composed of a 
set $P$ of planar points (each corresponding to a topological feature) above
the line $Y=X$, as well as the line $Y=X$ itself,
i.e., ${\cal P}=P\cup\{(x,y)|y=x\}$. Given a set of persistence diagrams
${\cal P}_1,...,{\cal P}_m$, for the data reduction purpose, one way to
summarize their topological features is to compute the {\em center} ${\cal C}$
of them first under the bottleneck distance. Here we mainly focus
on the two discrete versions when points in
${\cal C}$ could be selected with or without replacement from $P_i$'s. (We will
briefly discuss the continuous case, i.e., points in ${\cal C}$ are arbitrary,
which turns out to be closely related to the 3-dimensional geometric assignment
problem).
For technical reasons, we first focus on the
case when $|P_i|$'s are all the same (i.e., all have the same size $n$), and
the problem is to compute a center point set $C$ under the bottleneck matching
distance. We show, by a non-trivial reduction from the Planar 3D-Matching
problem, that this problem is NP-hard even when $m=3$ diagrams are given. This implies that the
general center problem for persistence diagrams under the bottleneck distance, when $P_i$'s possibly have different
sizes, is also NP-hard when $m\geq 3$. On the positive side, we show that this
problem is polynomially solvable when $m=2$ and admits a factor-2 approximation
for $m\geq 3$. These positive results hold for any $L_p$ metric when $P_i$'s
are point sets of the same size, and also hold for the case when $P_i$'s have
different sizes in the $L_\infty$ metric (i.e., for the Center Persistence
Diagram problem). This is the best possible in polynomial time for the
Center Persistence Diagram under the bottleneck distance unless P = NP.
All these results hold for both of the discrete versions as well as the continuous version; in fact, the
NP-hardness and approximation results also hold under the Wasserstein distance
for the continuous version.
\end{abstract}

\section{Introduction}

Computational topology has found a lot of applications in recent years
\cite{EH10}. Among them, persistence diagrams, each being a set of
(topological feature) points above and inclusive of the line $Y=X$ in the X-Y plane, have also
found various applications, for instance in GIS \cite{Ahmed14},
in neural science \cite{Giusti15}, in wireless networks \cite{Le15}, and
in prostate cancer research \cite{Lawson19}.
(Such a topological feature point $(b,d)$ in a persistence diagram, which we
will simply call a point henceforth, indicates a topological
feature which appears at time $b$ and disappears at time $d$. Hence $b\leq d$.
In the next section, we will present some technical details.)
A consequence is that practitioners gradually have a database of persistence
diagrams when processing the input data over certain period of time. It is not uncommon these days that such a database has tens of thousands
of persistence diagrams, each with up to several thousands of points.
How to process and search these diagrams becomes a new challenge for algorithm
designers, especially because the bottleneck distance is typically used to measure
the similarity between two persistence diagrams.

In \cite{Fasy18} the following problem was studied: given a set of persistence
diagrams ${\cal P}_1,...,{\cal P}_m$, each with size at most $n$, how to
preprocess them so that each has a key $k_i$ for $i=1,...,m$ and for a query
persistence diagram ${\cal Q}$ with key $k$, an approximate nearest persistence
diagram ${\cal P}_j$ can be returned by searching the key $k$ in the 
data structure for $k_i$'s. A hierarchical data structure was built and
the keys are basically constructed using snap roundings on a grid with
different resolutions. There is a trade-off between the space complexity (i.e., number
of keys stored) and the query time. For instance, if one wants an efficient
(polylogarithmic) query time, then he/she has to use an exponential
space; and with a linear or polynomial space, then he/she needs to 
spend an exponential query time. Different from traditional problems of searching
similar point sets \cite{HS94}, one of the main technical difficulties is to
handle points near the line $Y=X$.

In prostate cancer research, one important part is to determine how the cancer
progresses over certain period of time. In \cite{Lawson19}, a method is to
use a persistence diagram for each of the histopathology images (taken over
certain period of time). Naturally, for the data collected over some time
interval, one could consider packing a corresponding set of persistence
diagrams with a center, which could be considered as a {\em median}
persistence diagram summarizing these persistence diagrams. A sequence of
such centers over a longer time period would give a rough estimate on how the
prostate cancer progresses. This motivates our research. On the other hand,
while the traditional center concept (and the corresponding algorithms) has
been used for planar point sets (under the Euclidean distance) \cite{MIH81}
and on binary strings (under the Hamming distance) \cite{LMW02}; recently we
have also seen its applications in more complex objects, like polygonal chains
(under the discrete Frechet distance) \cite{Indyk02,Buchin19}. In this sense,
this paper is also along this line.

Formally, in this paper we consider a way to pack persistence diagrams. Namely,
given a set of persistence diagrams ${\cal P}_1,...,{\cal P}_m$, how to compute
a {\em center} persistence diagram? Here the distance measure used is the
traditional bottleneck distance and Wasserstein distance (where we first focus
on the former). We first describe the case when all ${\cal P}_i$'s have the
same size $n$, and later we show how to withdraw this constraint (by slightly
increasing the running time of the algorithms). It turns out that the
{\em continuous} case, i.e., when the points in the center can be arbitrary,
is very similar to the geometric 3-dimensional assignment problem: Given three
points sets $P_i$ of the same size $n$ and each colored with {\em color}-$i$
for $i=1..3$, divide points in $P_i$'s into $n$ 3-clusters (or triangles) such
that points in each cluster or triangle have different colors, and some
geometric quantity (like the maximum area or perimeter of these triangles) is
minimized \cite{Spieksma96,Goossens10}. For our application, we need to
investigate discrete versions where points in the center persistence diagram
must come from the input diagrams (might be from more than one diagrams). We
show that the problem is NP-hard even when $m=3$ diagrams are given. On the
other hand, we show that the problem is polynomially solvable when $m=2$ and
the problem admits a 2-approximation for $m\geq 3$. At the end, we briefly
discuss how to adapt the results to Wasserstein distance for the continuous
case. The following table summarizes the main results in this paper.

\begin{table}[ht]
\centering 
\caption{Results for the Center Persistence Diagram problems under the bottleneck ($d_B)$ and Wasserstein ($W_p$) distances when $m\geq 3$ diagrams are given.}

\begin{tabular}{l c c c }
\hline\hline
 ~~~ & ~~Hardness  &~~Inapproximability bound~ & ~~Approximation factor   \\
\hline
$d_B$, with no replacement  & NP-complete      & $2-\varepsilon$  &  2   \\
$d_B$, with replacement  & NP-complete     & $2-\varepsilon$  &  2   \\
$d_B$, continuous  & NP-hard      & $2-\varepsilon$  &  2   \\
 \hline
$W_p$, with no replacement  & ?      & ?  &  2   \\
$W_p$, with replacement  & ?      & ?  &  2   \\
$W_p$, continuous  & NP-hard      & ?  &  2   \\
 \hline
\end{tabular}
\end{table}

This paper is organized as follows. In Section 2, we give some necessary
definitions and we also show, as a warm-up, that the case when $m=2$
is polynomially solvable. In Section 3, we prove that the Center Persistence
Diagram problem under the bottleneck distance is
NP-hard when $m=3$ via a non-trivial reduction from the Planar
three-dimensional Matching (Planar 3DM) problem. In Section 4, we present
the factor-2 approximation algorithm for the problem (when $m\geq 3$).
In Section 5, we briefly discuss how to modify the proofs to the continuous
Center Persistence Diagram under the Wasserstein distance.
In Section 6, we conclude the paper with some open questions.

\section{Preliminaries}

We assume that the readers are familiar with standard terms in algorithms,
like approximation algorithms \cite{CLRS01}, and  NP-completeness \cite{GJ79}.

\subsection{Persistence Diagram}

Homology is a machinery from algebraic topology which gives the ability to count the number of holes in
a $k$-dimensional simplicial complex. For instance, let $X$ be a simplicial 
complex, and let the corresponding $k$-dimensional homology be $H_k(X)$, then
the dimension of $H_0(X)$ is the number of path connected
components of $X$ and $H_1(X)$ consists of loops in $X$, each is a `hole' in
$X$. It is clear that these numbers are invariant under rigid motions (and
almost invariant under small numerical perturbations) on the original data,
which
is important in many applications. For further details on classical homology
theory, the readers are referred to \cite{TH01}, and to \cite{EH10,MKM04} for
additional information on computational homology. It is well known that the
$k$-dimensional homology of $X$ can be computed in polynomial
time \cite{EH10,ELZ02,MKM04}.

Ignoring the details for topology, the central part of persistent homology
is to track the birth and death of the topological features when computing
$H_k(X)$. These features give a {\em persistence diagram} (containing the
birth and death times of features as pairs $(b,d)$ in the extended plane).
See Figure 1 for an example. Note that as a convention, the line $Y=X$
is included in each persistence diagram, where points on the
line $Y=X$ provide infinite multiplicity, i.e., a point $(t,t)$ on it could
be considered as a dummy feature which is born at time $t$ then immediately
dies. Formally, a persistence diagram ${\cal P}$ is composed of a 
set $P$ of planar points (each corresponding to a topological feature) above
the line $Y=X$, as well as the line $Y=X$ itself,
i.e., ${\cal P}=P\cup\{(x,y)|y=x\}$. Due to the infinite multiplicity on $Y=X$,
there is always a bijection between ${\cal P}_i$ and
${\cal P}_j$, even if $P_i$ and $P_j$ have different sizes.

\begin{figure}[htbp]
%\psfrag{x1}{$x_1$}
%\psfrag{x2}{$x_2$}
%\psfrag{x3}{$x_3$}
\begin{center}
\includegraphics[bb=-80 10 400 230,totalheight=5.0cm]{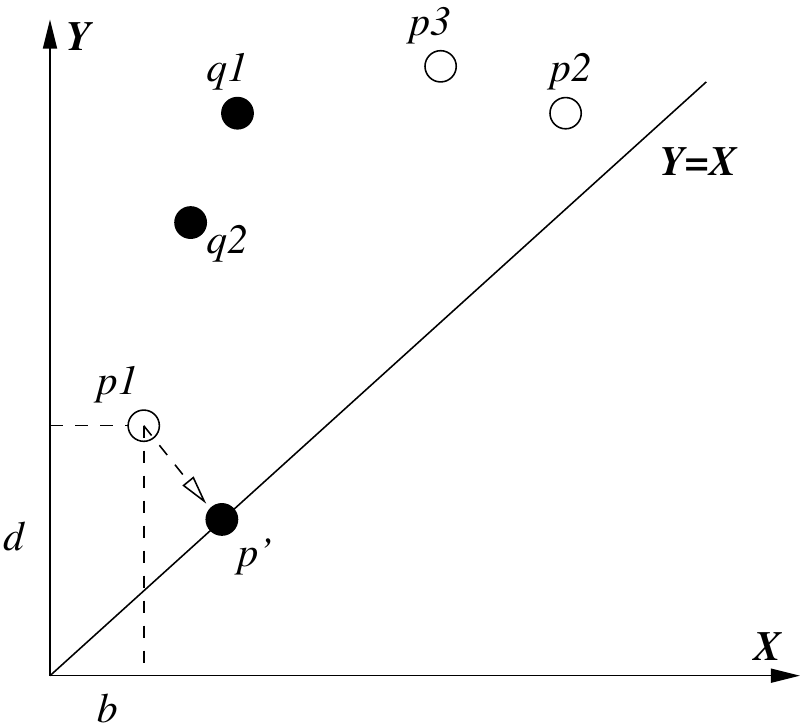}
\end{center}
\label{fig0}
\caption{\bf Two persistence diagrams ${\cal P}$ and ${\cal Q}$, with feature point sets $P=\{p_1,p_2,p_3\}$ and $Q=\{q_1,q_2\}$ respectively. A point $p_1=(b,d)$ means that it is born at time $b$ and it dies at time $d$. The projection of $p_1$ on $Y=X$ gives $p'$.}
\end{figure}
Given two persistent diagrams ${\cal P}_i$ and ${\cal P}_j$, each with $O(n)$ points,
the {\em bottleneck distance} between them is defined as follows:
$$d_B({\cal P}_i,{\cal P}_j)=\inf_{\phi}\{ \sup_{x\in {\cal P}_i}\|x-\phi(x)\|_{\infty}, \phi:{\cal P}_i\rightarrow {\cal P}_j \mbox{~is~a~bijection}\}.$$
Similarly, the $p$-{\em Wasserstein} distance is defined as
$$W_p({\cal P}_i,{\cal P}_j)=\left( \inf_{\phi}\sum_{x\in {\cal P}_i}\|x-\phi(x)\|^p_{\infty}\right)^{1/p}, \phi:{\cal P}_i\rightarrow {\cal P}_j \mbox{~is~a~bijection}.$$
We refer the readers to \cite{EH10} for further information regarding persistence diagrams. Our extended results regarding Wasserstein distance will be
discussed solely in Section 5, and until then we focus only on the bottleneck
distance.
 
For point sets $P_1,P_2$ of the same size, we will also use $d^p_B(P_1,P_2)$ to
represent their bottleneck matching distance, i.e., let $\beta$ be a bijection
between $P_1$ and $P_2$,
$$d^p_B(P_1,P_2)=\min_{\beta}\max_{a\in P_1} d_p(a,\beta(a)).$$ 
Here, $d_p(-)$ is the distance under the $L_p$ metric. As we mainly cover the
case $p=2$, we will use $d_B(P_i,P_j)$ instead of $d^2_B(P_i,P_j)$ henceforth.
Note that in comparing persistence diagrams, the $L_\infty$ metric is always
used. For our hardness constructions, all the valid clusters form either
horizontal or vertical segments, hence the distances under $L_2$ and $L_\infty$
metrics are all equal in our constructions.

While the bottleneck distance between two persistence diagrams is continuous
in its original form, it was shown that it can be computed using
a discrete method \cite{EH10}, i.e., the traditional geometric bottleneck
matching \cite{Alon01}, in $O(n^{1.5}\log n)$ time.
In fact, it was shown that the multiplicity property of the line $Y=X$ can be
used to compute the bottleneck matching between two
diagrams ${\cal P}_1$ and ${\cal P}_2$ more conveniently --- regardless of
their sizes \cite{EH10}. This can be done as follows. Let $P_i$ be the set
of feature points in ${\cal P}_i$. Then project points in $P_i$ perpendicularly
on $Y=X$ to have $P'_i$ respectively, for $i=1,2$. (See also Figure 1.) It
was shown that the bottleneck distance between two diagrams ${\cal P}_1$ and
${\cal P}_2$  is exactly equal to the bottleneck (bipartite) matching distance, in the
$L_\infty$ metric, between $P_1\cup P'_2$ and $P_2\cup P'_1$. Here the weight
or cost of an edge $c(u,v)$, with $u\in P'_2$ and $v\in P'_1$, is set to zero;
while $c(u,v)=\|u-v\|_\infty$, if $u\in P_1$ or $v\in P_2$. The $p$-Wasserstein
distance can be computed similarly, using a min-sum bipartite matching between
$P_1\cup P'_2$ and $P_2\cup P'_1$, with all edge costs raised to $c^p$.
(Kerber, et al. showed that several steps of the bottleneck matching
algorithm can be further simplified \cite{Kerber16}.)
Later, we will extend this construction for more than two diagrams.

\subsection{Problem Definition}

Throughout this paper, for two points $p_1=(x_1,y_1)$ and
$p_2=(x_2,y_2)$, we use $d_p(p_1,p_2)$ to represent the $L_p$
distance between $p_1$ and $p_2$, which is
$d_p(p_1,p_2)=(|x_1-x_2|^p+|y_1-y_2|^p)^{1/p},$ for $p<\infty$.
When $p=\infty$, $d_\infty(p_1,p_2)=\|p_1-p_2\|_\infty=\max\{|x_1-x_2|,|y_1-y_2|\}$.
We will mainly focus on $L_2$ and $L_\infty$ metrics, for the former,
we simplify it as $d(p_1,p_2)$.

\begin{definition}\textbf{\emph{The Center Persistence Diagram Problem under the Bottleneck Distance (CPD-B)}}

{\bf Instance}: A set of $m$ persistence diagrams ${\cal P}_1,...,{\cal P}_m$
with the corresponding feature point sets $P_1,...,P_m$ respectively, and a
real value $r$.

{\bf Question}: Is there a persistence diagram ${\cal Q}$ such that $\max_i d_B({\cal Q},{\cal P}_i)\leq r$?
\end{definition}

Note that we could have three versions, depending on ${\cal Q}$. We mainly
focus on the discrete version when the points in ${\cal Q}$ are selected with
no replacement from the multiset $\cup_{i=1..m} P_i$. It turns out that the
other discrete version, i.e., the points in ${\cal Q}$ are selected with
replacement from the set $\cup_{i=1..m} P_i$, is different from the first
version but all the results can be carried over with some simple twist. 
We will briefly cover the {\em continuous} case, i.e., when points ${\cal Q}$
are arbitrary; as we covered earlier in the introduction, when $m=3$, this
version is very similar to the geometric three-dimensional assignment
problem \cite{Spieksma96,Goossens10}. 

We will firstly consider two simplified versions of the corresponding problem.

\begin{definition}\textbf{\emph{The $m$-Bottleneck Matching Without Replacement Problem}}

{\bf Instance}: A set of $m$ planar point sets $P_1,...,P_m$ such that
$|P_1|=\cdots=|P_m|=n$, and a real value $r$.

{\bf Question}: Is there a point set $Q$, with $|Q|=n$, such that any
$q\in Q$ is selected from the multiset $\cup_i P_i$ with
no replacement and $\max_i d_B(Q,P_i)\leq r$?
\end{definition}
  
\begin{definition}\textbf{\emph{The $m$-Bottleneck Matching With Replacement Problem}}

{\bf Instance}: A set of $m$ planar point sets $P_1,...,P_m$ such that
$|P_1|=\cdots=|P_m|=n$, and a real value $r$.

{\bf Question}: Is there a point set $Q$, with $|Q|=n$, such that any
$q\in Q$ is selected from the set $\cup_i P_i$ with
replacement and $\max_i d_B(Q,P_i)\leq r$?
\end{definition}
  
\begin{figure}[htbp]
%\psfrag{x1}{$x_1$}
%\psfrag{x2}{$x_2$}
%\psfrag{x3}{$x_3$}
\begin{center}
\includegraphics[bb=-160 20 380 160,totalheight=3.8cm]{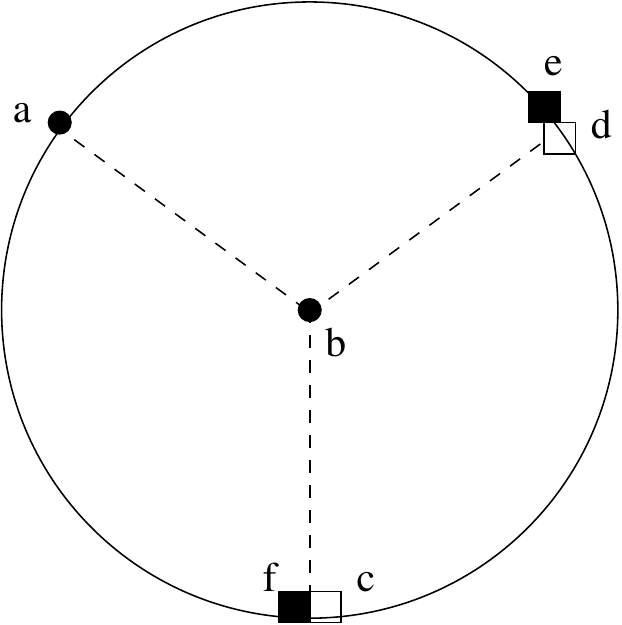}
\end{center}
\label{fig0}
\caption{\bf An example with $P_1=\{a,b\}, P_2=\{c,d\}$ and $P_3=\{e,f\}$,
with all the points (except $b$) on a unit circle and $b$ being the center
of the circle. For the `without replacement' version, the optimal solution is
$Q_1=\{b,c\}$, where $b$ covers the 3-cluster $\{a,d,e\}$, $c$ covers
the 3-cluster $\{b,c,f\}$ and the optimal covering radius is 1. 
For the `with replacement' version, the optimal solution could be the same, but
could also be $\{b,b\}$.}
\end{figure}
It turns out that these two problems are really to find center points in $Q$ to
cover $m$-clusters with an optimal covering radius $r$, with each cluster being
composed of $m$ points, one each from $P_i$. For $m=3$, this is similar to the
geometric three-dimensional assignment problem which aims at finding
$m$-clusters with certain criteria \cite{Spieksma96,Goossens10}. 
However, the two versions of the problem are slightly different 
from the geometric three-dimensional assignment problem.
The main difference is that in these discrete versions a cluster could
be covered by a center point which does not belong to the cluster. See
Figure 2 for an example.
Also, note that the two versions themselves are slightly different; in fact,
their solution values could differ by a factor of 2 (see Figure 3).

Note that we could define a continuous version in which the condition on $q$
is withdrawn and this will be briefly covered at the end of each section.
In fact, we focus more on the optimization versions of these problems.
We will show that 3-Bottleneck Matching, for both the discrete versions,
is NP-hard, immediately implying CPD-B is NP-hard
for $m\geq 3$. We then present a 2-approximation for the $m$-Bottleneck
Matching Problem and later we will
show how to make some simple generalization so the `equal size' condition
can be withdrawn for persistence diagrams --- this implies that CPD-B also
admits a 2-approximation for $m\geq 3$. We will focus on the `without
replacement' version in our writing, and later we will show how to
generalize it to the `with replacement' version at the end of each section.
Henceforth, we will refer to the 'without replacement' version simply as
$m$-Bottleneck Matching unless otherwise specified.

At first, we briefly go over a polynomial time solution for the case when $m=2$.

\begin{figure}[htbp]
%\psfrag{x1}{$x_1$}
%\psfrag{x2}{$x_2$}
%\psfrag{x3}{$x_3$}
\begin{center}
\includegraphics[bb=150 240 340 300,totalheight=1.7cm]{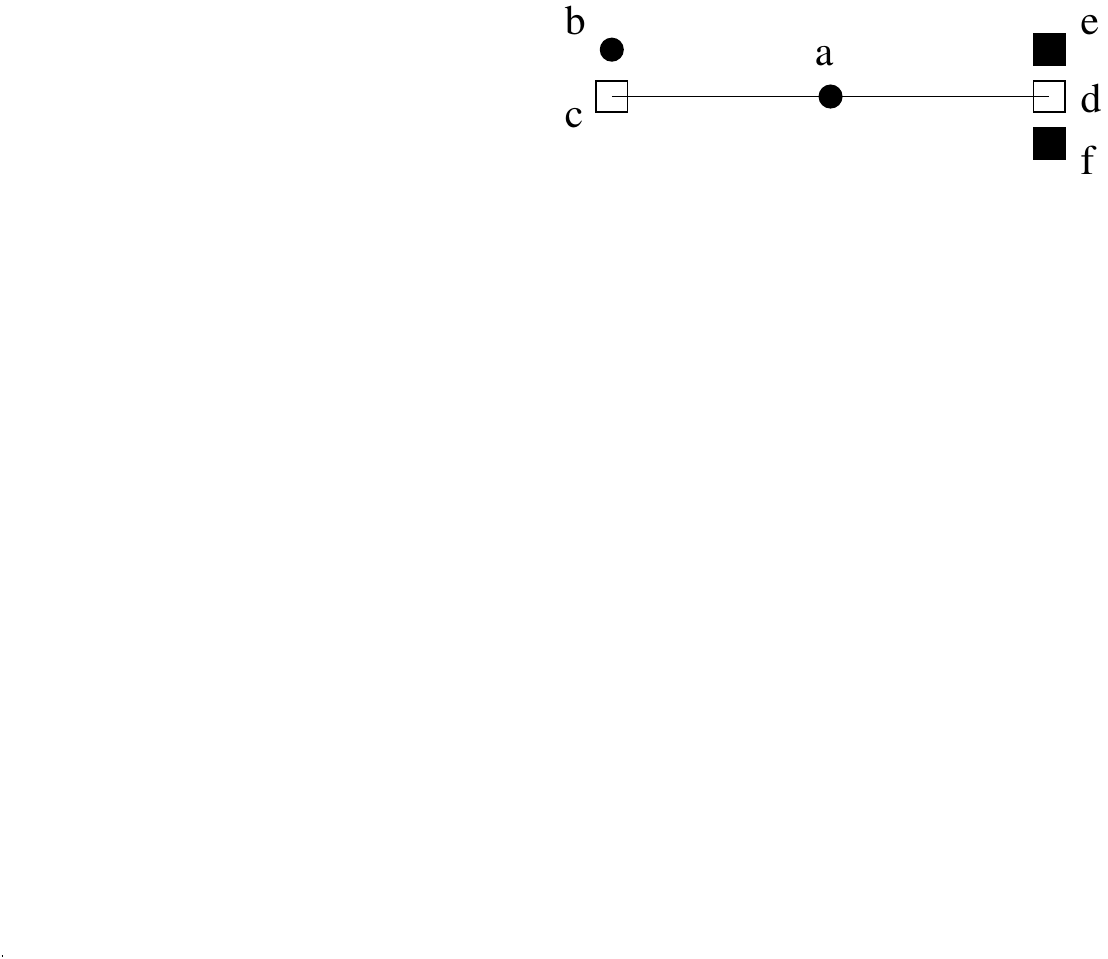}
\end{center}
\label{fig0}
\caption{\bf An example with $P_1=\{a,b\}, P_2=\{c,d\}$ and $P_3=\{e,f\}$,
with all the points on a unit line segment and $a$ being the midpoint of the segment. For the `without replacement' version, the optimal solution is
$Q_1=\{a,b\}$, where $a$ covers the 3-cluster $\{a,c,f\}$, $b$ covers
the 3-cluster $\{b,d,e\}$ and the optimal covering radius is 1. For the
`with replacement' version, the optimal solution is $Q_2=\{a,a\}$, with the
same clusters $\{a,c,f\}$ and $\{b,d,e\}$, and the optimal covering radius
being 1/2.}
\end{figure}

\subsection{A Warm-up for $m=2$}

First, recall that $|P_1|=|P_2|=n$.
Note that the optimal solution must be the distance between $p\in P_1$ and
$q\in P_2$. We first consider the decision version of the problem; namely,
given a radius $r$, find a clustering of 2-points $\{p_{1i},p_{2j}\}$, with
$p_{xy}\in P_x (x=1,2)$, such that the distance between each of them and $\hat{q}$, where
$\hat{q}$ is selected with no replacement from the multiset $P_1\cup P_2$, is at most $r$.
(Note that $\hat{q}$ is not necessarily equal to $p_{1i}$ or $p_{2j}$.)
Once having this decision procedure, we could use a binary search to compute the
smallest radius such a clustering exists.

Given the radius $r$, we construct a flow network $G = (V, E)$ as follows:
besides the source $s$ and the sink $t$, there are four layers of nodes.
The first layer contains $n$ nodes corresponding to the $n$ points of $P_1$,
the second and the third layers both contain $2n$ nodes corresponding to the $2n$ points of $P_1 \cup P_2$,
and the fourth layer contains $n$ nodes corresponding to the $n$ points of $P_2$.
Each node $p_{1i}$ in the first layer has a link going out to a node $p_{\ell j}$ in the second layer if their distance is at most $r$;
for example, $p_{1i}$ in the first layer surely has a link going out to $p_{1i}$ in the second layer, due to their distance being $0$.
Each node $p_{\ell i}$ in the second layer has only one out-going link to the node $p_{\ell i}$ in the third layer.
Each node $p_{\ell i}$ in the third layer has a link going out to a node $p_{2 j}$ in the fourth layer if their distance is at most $r$.
Lastly, the source $s$ has a link going out to every node in the first layer,
and every node in the fourth layer has a link going out to the sink $t$.
All the links in the constructed network has unit capacity.

One sees that the path $s$-$p_{1h}$-$p_{\ell i}$-$p_{\ell i}$-$p_{2 j}$ is used to send a unit of flow if and only if
1) $p_{\ell i}$ is selected into the set $Q$, and
2) $p_{\ell i}$ is matched with $p_{1h}$ ($p_{2j}$, respectively) in the bottleneck matching between $Q$ and $P_1$ ($P_2$, respectively).
Therefore, the maximum flow has a value $n$ if and only if $Q$ is determined such that the maximum bottleneck matching distance is at most $r$.
Since the constructed network is acyclic, its maximum flow can be computed in $O(n^3)$ time \cite{MKM78}. A binary search to find the optimal radius leads to
a solution running in $O(n^3\log n)$ time.

For the `With Replacement' version, given a radius $r$, an unweighted
bipartite graph between $P_1$ and $P_2$ can be first constructed. In the graph
there is an edge between $u\in P_1$ and $v\in P_2$ if there is a point
$w\in P_1\cup P_2$ such that a circle with radius $r$ centered $w$ can cover
both $u$ and $v$. Then, the decision problem is to decide whether a perfect
matching exists, which can be solved in $O(n^{2.5})$ time \cite{HK73}. A
binary search for the optimal radius gives a solution which runs in
$O(n^{2.5}\log n)$ time.

For the continuous version, the problem can be solved by computing the
geometric bottleneck matching between $P_1$ and $P_2$ in $O(n^{1.5}\log n)$
time \cite{Alon01}. After the matched edges between $P_1$ and $P_2$ are
identified, the set of midpoints of ll the edges in the matching gives
us the solution.

The above algorithm can be generalized for two persistence diagrams, with
the distance between two points in the $L_\infty$ metric.
In general, the sizes of two persistence diagrams might not be the same, i.e.,
$|P_1|$ might not be the same as $|P_2|$. This can be handled easily using
the projection method in \cite{EH10}, which has been described in
subsection 2.1 and will be generalized for $m\geq 3$
in Section 4. Hence, we have the following theorem.

\begin{theorem}
%\label{thm01}
The Center Persistence Diagram Problem can be solved in polynomial time, for
$m=2$ and for all the three versions (`Without Replacement', `With Replacement'
and continuous versions).
\end{theorem}

In the next section, we will consider the case for $m=3$.

\section{3-Bottleneck Matching is NP-complete}

We will first focus on the $L_2$ metric in this section and at the
end of the proof it should be seen that the proof also works for the
$L_\infty$ metric. For $m=3$, we can color points in $P_1,P_2$ and
$P_3$ in color-1, color-2 and color-3. Then, in this case, the problem is really to
find $n$ disks centered at $n$ points from
$P_1\cup P_2\cup P_3$, with smallest radii $r^{*}_i$ ($i=1..n$) respectively, such that each disk contains exact 3 points of
different colors (possibly including the center of the disk); moreover, $\max_{i=1..n}r^{*}_i$ is bounded from above by a given value $r$.
We also say that these 3 points form a {\em cluster}.

It is easily seen that (the decision version of) 3-Bottleneck Matching is in NP. Once the $n$ guessed
disks are given, the problem is then a max-flow problem, which can be
verified in polynomial time.

We next show that Planar 3-D Matching (Planar 3DM) can be reduced to
3-Bottleneck Matching in polynomial time. The former is a known NP-complete
problem \cite{DF86}. In 3DM, we are given three sets of elements
$E_1,E_2,E_3$ (with $|E_1|=|E_2|=|E_3|=\gamma$) and a set ${\cal T}$ of $n$ triples,
where $T\in {\cal T}$ implies that $T=(a_1,a_2,a_3)$ with $a_i\in E_i$. The problem
is to decide whether there is a set $S$ of $\gamma$ triples such that each element
in $E_i$ appears exactly once in (the triples of) $S$. The Planar 3DM incurs 
an additional constraint: if we embed elements and triples as points on the
plane such that there is an edge between an element $a$ and a triple $T$ iff
$a$ appears in $T$, then the resulting graph is planar.

An example for Planar 3DM is as follows: $E_1=\{1,2\}, E_2=\{a,b\}, E_3=\{x,y\}$, and ${\cal T}=\{(1,a,x),(2,b,x),(2,b,y),(1,b,y)\}$. The solution is
$S=\{(1,a,x),(2,b,y)\}$. 

Given an instance for Planar 3DM and a corresponding planar graph $G$ with $O(n)$ vertices, we
first convert it to a planar graph with degree at most 3. This can be
done by replacing a degree-$d$ element node $x$ in $G$ with a path of $d$ nodes
$x_1,...,x_d$, each with degree at most 3 and the connection between
$x$ and a triple node $T$ is replaced by a connection from $x_i$ to $T$ for some $i$
(see also Figure 4). We have a resulting planar graph $G'=(V(G'),E(G'))$ with degree at
most 3 and with $O(n)$ vertices. Then we construct a rectilinear embedding of $G'$ on a regular rectilinear grid with a unit grid length,
where each vertex $u\in V(G')$ is embedded at a grid point and an edge $(u,v)\in E(G')$
is embedded as an intersection-free path between $u$ and $v$ on the grid.
It is well-known that such an embedding can be computed 
in $O(n^2)$ time \cite{Valiant81}.

Let $x$ be a black node ($\bullet$ in Figure~4) with degree
$d$ in $G$. In the rectilinear embedding of $G'$, the paths from $x_i$ to
$x_{i+1}$ ($i=1,...,d-1$) will be the basis of the element gadget for $x$.
(Henceforth, unless otherwise specified, everything we talk about in this
reduction is referred to the rectilinear embedding of $G'$.)
We put a copy of $\bullet$ at each (grid point) $x_i$ as in Figure~4. (If the
path from $x_i$ to $x_{i+1}$ is of length greater than one, then we put
$\bullet$ at each grid point on the path from $x_i$ to $x_{i+1}$.) 

We now put color-2 and color-3 points ($\square$ and $\blacksquare$) at 1/3
and 2/3 positions at each grid edge which is contained in some path in an
element gadget (in the embedding of $G'$). These points are put in a way such
that it is impossible to use a discrete disk centered at a $\bullet$ point
with radius 1/3 to cover three points with different colors. These patterns
are repeated to reach a \emph{triple gadget}, which will be given later.
Note that this construction is done similarly for elements $y$ and $z$, except
that the grid points in the element gadgets for $y$ and $z$ are of color-2
($\square$) and color-3 ($\blacksquare$) respectively.

\begin{lemma}
In an element gadget for $x$, exactly one $x_i$ is covered by a discrete
disk of radius 1/3, centered at a (colored) grid point out of the gadget.
\end{lemma}

\begin{proof}
Throughout the proof, we refer to Figure~4. Let $x$ be colored by
color-1 (e.g., $\bullet$). In the rectilinear embedding, let the path length
between $x_1$ and $x_d$ be $D$. Then, the total number of points on the path
from $x_1$ to $x_d$, of colors 1, 2 and 3, is $3D+1$.
By the placement of color-2 and color-3 points in the gadget for $x$,
exactly $3D$ points of them can be covered by $D$ discrete disks of radii
1/3 (centered either at color-2 or color-3 points in the gadget).
Therefore, exactly one of $x_i$ must be covered by a discrete disk centered at a point out of the gadget.
\end{proof}

When $x_i$ is covered by a discrete disk of radius 1/3 centered at a point out of the gadget
$x$, we also say that $x_i$ is {\em pulled out} of $x$.

\begin{figure}[htbp]
%\psfrag{x1}{$x_1$}
%\psfrag{x2}{$x_2$}
%\psfrag{x3}{$x_3$}
\begin{center}
\includegraphics[bb=-40 0 460 110,totalheight=3.4cm]{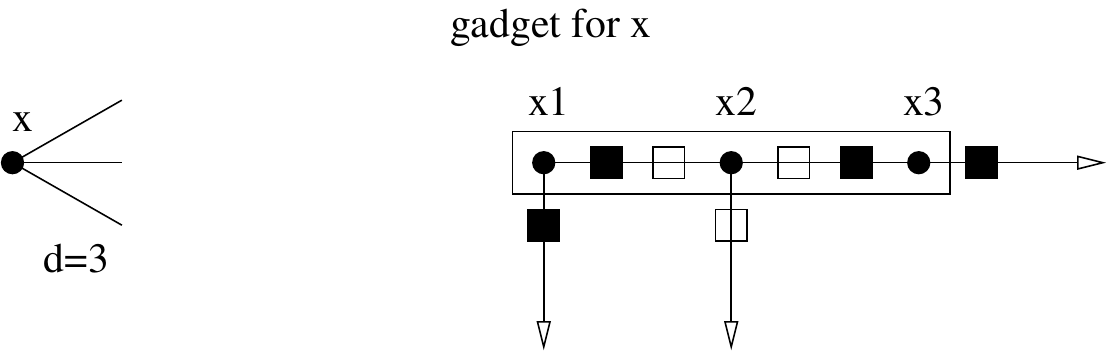}
\end{center}
\label{fig2}
\caption{\bf The gadget for element $x$.}
\end{figure}

We now illustrate how to construct a triple gadget $T=(x,y,z)$. It is
basically a grid point on which we put three points with different
colors. (In Figure~5, we simply use a $\blacktriangle$
representing such a triple gadget.) The interpretation of $T$ being
selected in a solution for Planar 3DM is that the three colored points
at $\blacktriangle$ is covered by a disk of radius zero, centered at one of these
three points. When one of these three points at
$\blacktriangle$ is covered by a disk of radius 1/3 centered at some other points 
(on the path from one of the elements $x,y$ or $z$ to $T$), we say that
such a point is {\em pulled out} of the triple gadget $T$ by the corresponding element gadget.

\begin{lemma}
In a triple gadget for $T=(x,y,z)$, to cover the three points representing $T$ using discrete disks of radii at most 1/3,
either all the three points are pulled out of the triple gadget $T$ by the three element gadgets respectively,
or none is pulled out. In the latter case, these three points can be covered by a discrete disk of radius zero.
\end{lemma}

\begin{proof}
Throughout the proof, we refer to Figure~5.
At the triple gadget $T$, if only one point (say $\bullet$) is pulled out or
two points (say, $\bullet$ and $\square$) are pulled out, then the remaining
points in the triple, $\square$ and $\blacksquare$ or $\blacksquare$ respectively,
could not be properly covered by a discrete disk of radius 1/3 --- such a 
disk would not be able to cover a cluster of exactly three points of distinct
colors.
Therefore, either all the three points associated with $T$ are pulled out by 
the three corresponding element gadgets, hence covered by three different
discrete disks of radii 1/3; or none of these three points is pulled out.
Clearly, in the latter case, these three points associated with $T$
can be covered by a discrete disk of radius zero, as a cluster.
\end{proof}

In Figure~5, we show the case when $x$ would not pull any point out of the gadget for
$T$. By Lemma 5, $y$ and $z$ would do the same, leading
$T=\langle x,y,z\rangle$ to be selected in a solution $S$
for Planar 3DM.
Similarly, in Figure~6, $x$ would pull a $\bullet$ point out
of $T$. Again, by Lemma 5, $y$ and $z$ would pull $\square$ and $\blacksquare$
points (one each) out of $T$, which implies that $T$ would not be selected
in a solution $S$ for Planar 3DM.

\begin{figure}[htbp]
%\psfrag{x1}{$x_1$}
%\psfrag{x2}{$x_2$}
%\psfrag{x3}{$x_3$}
\begin{center}
\includegraphics[bb=0 0 400 160,totalheight=4.6cm]{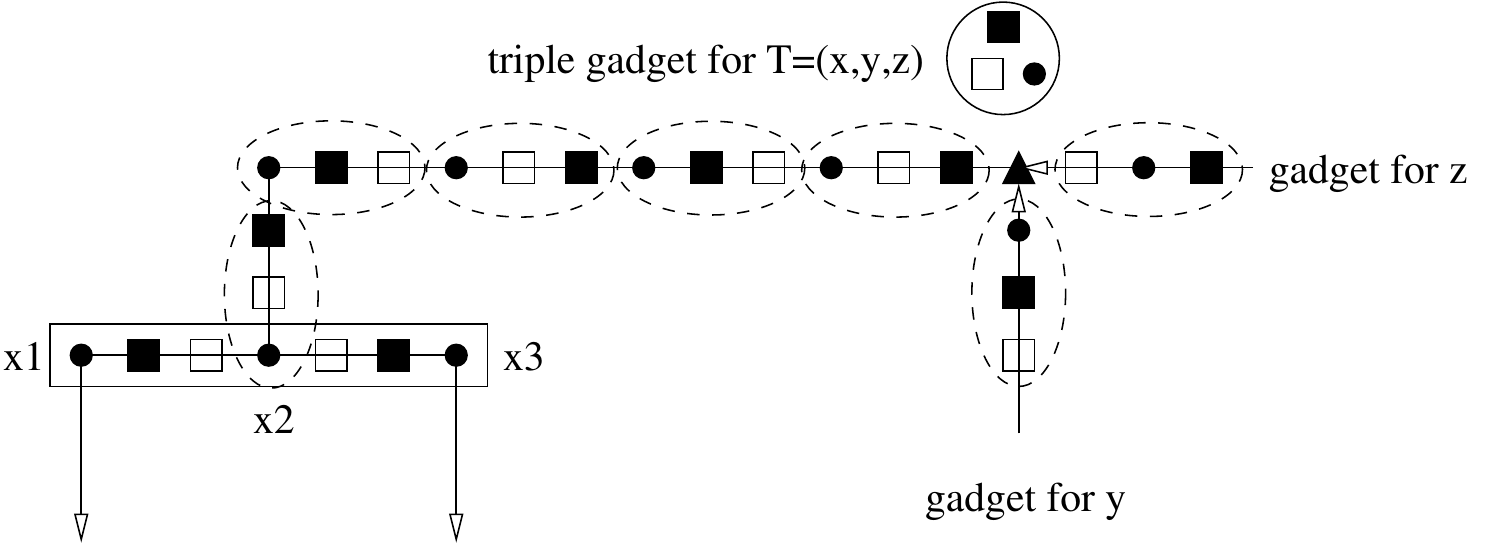}
\end{center}
\label{fig3}
\caption{\bf The triple gadget for $T=\langle x,y,x\rangle$ (represented as $\blacktriangle$, which is really putting three element points on a grid point). In this case the triple $\langle x,y,z\rangle$
is selected in the final solution (assuming operations are similarly performed on $y,z$). Exactly one of $x_i$ (in this case
$x_2$) is pulled out of the gadget for the element $x$.}
\end{figure}

\begin{figure}[htbp]
%\psfrag{x1}{$x_1$}
%\psfrag{x2}{$x_2$}
%\psfrag{x3}{$x_3$}
\begin{center}
\includegraphics[bb=0 0 400 160,totalheight=4.6cm]{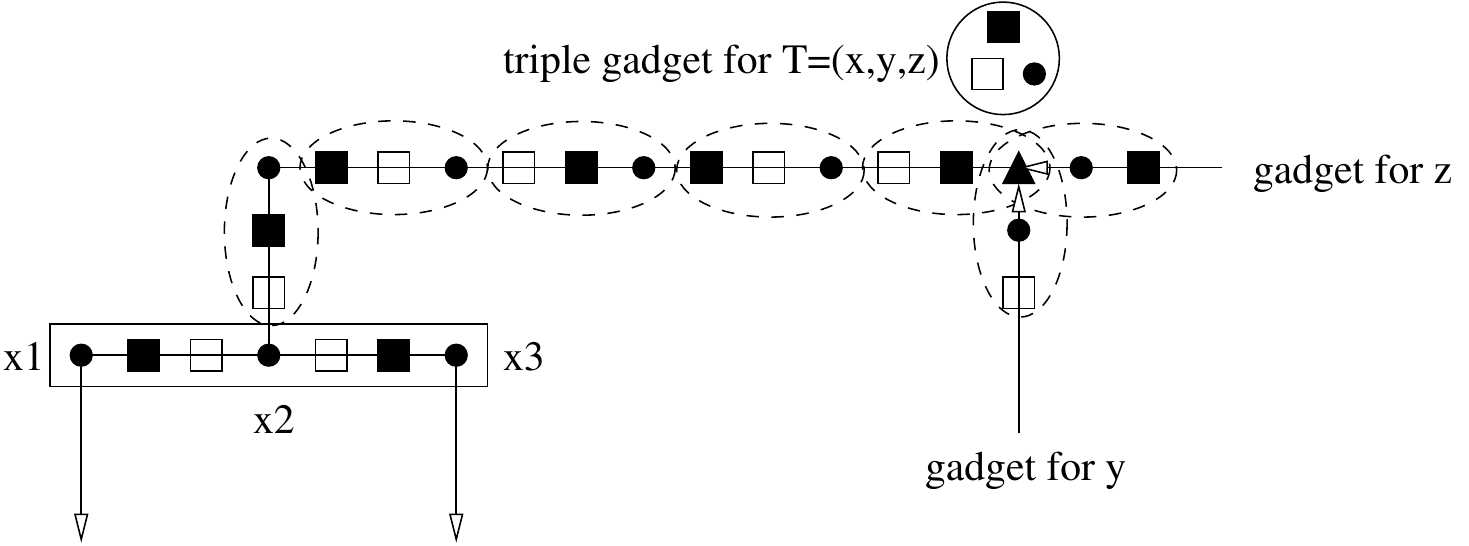}
\end{center}
\label{fig4}
\caption{\bf The triple gadget for $T=\langle x,y,x\rangle$ (represented as $\blacktriangle$, which is really putting three element points on a grid point). In this case the triple $\langle x,y,z\rangle$ would not be selected in the final solution. Note that the black round point in the triple gadget is pulled out by the element $x$, and the other two points are pulled out similarly by the element $y$ and $z$.}
\end{figure}

We hence have the following theorem.

\begin{theorem}
%\label{thm02}
The decision versions of 3-Bottleneck Matching for both the `Without
Replacement' and `With Replacement' cases are NP-complete, and the decision
version of the continuous 3-Bottleneck Matching is NP-hard.
\end{theorem}

\begin{proof}
We discuss the `Without Replacement' discrete case first, and will drop
the keyword `Without Replacement' until the end of the proof.
As explained a bit earlier, (the decision version of) 3-Bottleneck Matching is obviously in NP.
Moreover, we show in Lemma 5 and Lemma 6 that, given an instance for
Planar 3DM with $n$ triples over $3\gamma$ base elements we can convert
it into an instance $I$ of $3Kn$ points of three colors ($Kn$ points are
of color-1, color-2 and color-3 respectively) in polynomial time,
where $K$ is related to this polynomial running time.
We just formally argue below that Planar 3DM has a solution of $\gamma$ triples if
and only if
the converted instance $I$ of $3Kn$ points can be partitioned into
$Kn$ clusters each covered by a discrete disk of radius 1/3; moreover, there are
exactly $\gamma$ such clusters which are covered by discrete disks with radii zero.

`Only if part:' If the Planar 3DM instance has a solution, we have a set $S$ of
$\gamma$ triples which uniquely cover all the $3\gamma$ elements. Then, at each of the
corresponding $\gamma$ triple gadgets, we use a discrete disk of radius zero
to cover the corresponding three points. By Lemma 5,
in each element gadget $x$ exactly one point $x_i$ could be pulled out of the
gadget, connecting to these selected triple gadgets. By Lemma 6, the 
triple gadgets can be covered exactly in two ways. Hence the triples not
corresponding to $S$ will be covered in the other way: for
$T=\langle x,y,z\rangle$ not in $S$, one point of each color will be pulled
out of the triple gadget for $T$.

`If part:' If the converted instance $I$ of $3Kn$ points can be partitioned
into $Kn$ clusters each covered by a discrete disk of radius 1/3 and there are
exactly $\gamma$ clusters whose covering discrete disks have radii zero, then the triples
corresponding these clusters of point form a solution to the original
Planar 3DM instance. The reason is that, by Lemma 6, the remaining points
will be covered by a discrete disk of radius 1/3 (and cannot be further shrunk).
Moreover, by Lemma 5, at each element gadget, exactly one point will be
fulled out, leading to the corresponding triple gadget being covered by
a discrete disk of radius zero --- which implies that exactly one element
is covered by a selected triple.

It can be seen by now the proof works for the `With Replacement' discrete
version without any change in the proof. For the continuous version,
the NP-hardness holds with the same reduction. The reason for the NP-hardness
to hold is that the optimal grouping of three points (0,0), (0,1/3), (0,2/3)
is to use (0,1/3) as the continuous center, even though itself is still
discrete. However, the NP membership does not hold anymore for the
continuous case (since a guessed solution involve real numbers).
This closes our proof.
\end{proof}

Note that in the above proof, if Planar 3DM does not have a solution,
then we need to use discrete disks of radii at least 2/3 to have a valid 
solution for 3-Bottleneck Matching. This implies that finding a
factor-$(2-\varepsilon)$ approximation for (the optimization version of)
3-Bottleneck Matching remains NP-hard.

\begin{corollary}
It is NP-hard to approximate (the optimization version of) 3-Bottleneck Matching within a factor $2-\varepsilon$, for some $\varepsilon>0$ and for all the three versions.
\end{corollary}

We comment that the NP-hardness proofs in \cite{Goossens10,Spieksma96} also
use a reduction from Planar 3DM; however, those proofs are only for the $L_2$
metric. Here, it is clear that our reduction also works for the $L_\infty$
metric without any modification --- this is due to that all clusters in
our construction are either horizontal or vertical, therefore the distances
within a cluster would be the same under $L_2$ and $L_\infty$. With respect to
the CPD-B problem, points in color-$i$, $i=1,2,3$, are the basis for us to
construct a persistence diagram. To handle the line $Y=X$ in a persistence
diagram, let the diameter of the (union of the)
three constructed point sets of different colors be $\hat{D}$, we then
translate these points as a whole set rigidly such that all the points are at least
$2\hat{D}$ distance away from $Y=X$. We then have three persistence diagrams.
(The translation is to neutralize the infinite 
multiplicity of $Y=X$, i.e., to enforce that all points on $Y=X$ can be
ignored when computing the bottleneck distance between the corresponding
persistence diagrams.) Hence, we have the following corollary.

\begin{corollary}
It is NP-hard to approximate (the optimization version of) Center Persistence
Diagram problem under the bottleneck distance for $m\geq 3$ within a factor
$2-\varepsilon$, for some $\varepsilon>0$ and for all the three versions.
\end{corollary}

In the next section, we present tight approximation algorithms for the above
problems.

\section{A Tight Approximation}
\subsection{Approximation for $m$-Bottleneck Matching}
We first present a simple Algorithm 1 for $m$-Bottleneck Matching as follows.
Recall that in the $m$-Bottleneck Matching problem we are given $m$ sets of
planar points $P_1,...,P_{m}$, all with the same size $n$.
Without of generality, let the points in $P_i$ be colored with color-$i$.
\begin{enumerate}
\item Pick any color, say, color-$1$.
\item Compute the bottleneck matching $M_{1,i}$ between $P_1$ and $P_i$ for $i=2,...,m$.
\item For the $m-1$ edges $(p^{1}_{j_1},p^{i}_{j_i}) \in M_{1,i}$ for $i=2,...m$, where $p^{x}_{y} \in P_x$ for $x=1,...,m$, form a cluster $\{p^{1}_{j_1}, p^{2}_{j_2},...,p^{m}_{j_m}\}$ with $p^{1}_{j_1}$ as its center.
\end{enumerate}

We comment that the algorithm itself is similar to the one given for $m=3$ in
\cite{Goossens10}, which has a different objective function (i.e., minimizing the
maximum perimeter of clusters). We show next that Algorithm 1 is a factor-$2$
approximation for $m$-Bottleneck Matching. Surprisingly, the main tool here is
the triangle inequality of a distance measure. Note that we can not only
handle for any given $m\geq 3$, we also need some twist in the proof a bit
later for the three versions of the Center Persistence Diagram problem, where
the diagrams could have different sizes.

%We first present a simple Algorithm 1 for $m$-Bottleneck Matching as follows.
%\begin{enumerate}
%\item Pick any color-$i$, say, color-$1$.
%\item Compute the bottleneck matching $M_{1,2}$ between $P_1$ and $P_2$.
%\item Compute the bottleneck matching $M_{1,3}$ between $P_1$ and $P_3$.
%\item For any pair of edges $(p_{1,i}, p_{2,j}) \in M_{1,2}$ and $(p_{1,i}, p_{3,k}) \in M_{1,3}$, where $p_{x,y} \in P_x$ for $x=1, 2, 3$, form a cluster $\{p_{1,i}, p_{2,j}, p_{3,k}\}$ with $p_{1,i}$ as its discrete center.
%\end{enumerate}
%
%We show next that Algorithm 1 is a factor-$2$ approximation for $3$-Bottleneck Matching. We comment that the algorithm itself is very much the same as in
%\cite{Goossens10}, but the proof needs some twist --- especially for the
%`With Replacement' case (see Section 4.3).
\begin{theorem}
\label{thm09}
{\em Algorithm 1} is a polynomial time factor-$2$ approximation for $m$-Bottleneck Matching for all the three versions (i.e., `Without Replacement', `With Replacement' and continuous versions).
\end{theorem}

\begin{proof}
One clearly sees that the running time of Algorithm 1 is $O(mn^{1.5} \log n)$.

(1) We discuss the `Without Replacement' first.
In an optimal solution for $m$-Bottleneck Matching with its radius $\OPT$,
let $\{q_1, q_2,...,, q_m\}$ denote a cluster with its discrete center $\hat{q}$,
where $q_i$ is in color-$i$, for $i=1,...,m$.

From $\OPT \ge \max \{ d(\hat{q},q_1), d(\hat{q},q_2),..., d(\hat{q},q_m)\}$
and the triangle inequality,
we have $d(q_1, q_j) \le d(q_1, \hat{q}) + d(q_j, \hat{q}) \le 2 \cdot \OPT$
for $j=2,...,m$.
This implies a matching between $P_1$ and $P_i$ with radius at most $2 \cdot \OPT$, for $i=2,...,m$.

Let $\APP$ denote the maximum radius between the $m-1$ bottleneck matchings
computed in Algorithm 1; then the radius of the produced solution is $\APP$,
and we have $\APP \le d(q_1,q_j) \le 2 \cdot \OPT$ for $j=2,...,m$.
That is, Algorithm 1 is a polynomial time factor-$2$ approximation for
$m$-Bottleneck Matching (for the `Without Replacement' version).

(2) We next discuss the `With Replacement' case. Note that we never need to
change the algorithm.

In this case, first note that the points in the center ${\cal C}$ are selected
from $P_i$'s with replacement.
%, it can be easily seen that the problem
%can also be solved in polynomial time with a perfect matching method
%combined with a binary search for the radius when $m=2$, and the
%NP-hardness proof also works for $m\geq 3$.
% We show below that the
%2-approximation algorithm could still be applied to this case. This is done
%in the following, for $m=3$ under the Euclidean distance.
Then, given an optimal solution for this case we notice that some points in
$P_i$'s can be selected multiple times (i.e., more than once) in ${\cal C}$. Let $q$ be such
a point in ${\cal C}$. If $q$ covers a cluster including itself, then we
leave that cluster alone; otherwise, pick any cluster covered by $q$ and
also leave $q$ and that cluster alone. Then, anytime when $q$ covers
$\{p^{1}_{j_1}, p^{2}_{j_2},..., p^{m}_{j_m}\}$ once more with $q\not\in
\{p^{1}_{j_1}, p^{2}_{j_2},..., p^{m}_{j_m}\}$, we switch the center for
this cluster to $p^{1}_{j_1}$ (i.e., the point with color-1). Clearly, we have
$$d(p^{1}_{j_1},p^{i}_{j_i})\leq d(p^{1}_{j_1},q) + d(q,p^{i}_{j_i})\leq 2\cdot\OPT,$$
for $i=2,...,m$.
Then, combined with the other (`Without Replacement') case covered in part (1),
we can conclude that Algorithm~1 provides a 2-approximation for the `With Replacement'
case as well when $m\geq 3$. In fact, the example in Figure~3 shows a simple
matching lower bound of factor 2. 

(3) The proof for the continuous case would be almost identical as in part (1);
in fact, we just need to define $\hat{q}$ as an arbitrary point covering the
given cluster $\{q_1, q_2,...,, q_m\}$. And the remaining arguments would be
the same.
\end{proof}

\subsection{Generalization to the Center Persistence Diagram Problem under the Bottleneck Distance}

First of all, note that the above approximation algorithm works for
$m$-Bottleneck Matching when the metric is $L_\infty$. Hence, obviously it
works for the case when the input is a set of $m$ persistence diagrams (all
having the same size), whose (feature) points are all far away from $Y=X$,
and the distance measure is the bottleneck distance. (Recall that, when
computing the bottleneck distance between two persistence diagrams using a
projection method, we always use the $L_\infty$ metric to measure the
distance between two points.)

We next show how to generalize the factor-2 approximation algorithm for
$m$-Bottleneck Matching to the Center Persistence Diagram problem, first
for $m=3$. Note that we are given $m$ persistence diagrams ${\cal P}_1$,
${\cal P}_2,...,$ and ${\cal P}_m$, with the corresponding non-diagonal
point sets being $P_1,P_2,...,$ and $P_m$ respectively. Here the sizes of
$P_i$'s could be different and we assume that the points in $P_i$ are of
color-$i$ for $i=1,...,m$.

Given a point $p\in P_i$, let $\tau(p)$ be the (perpendicular) projection of
$p$ on the line $Y=X$. Consequently, let $\tau(P_i)$ be the projected points
of $P_i$ on $Y=X$, i.e.,
$$\tau(P_i)=\{\tau(p)|p\in P_i\}.$$

When $m=2$, i.e., when we are only given ${\cal P}_1$ and ${\cal P}_2$, not
necessarily of the same size, it was shown by Edelsbrunner and Harer that 
$d_B({\cal P}_1,{\cal P}_2)=d^\infty_B(P_1\cup \tau(P_2), P_2\cup \tau(P_1))$
\cite{EH10}. (Note that $|P_1\cup\tau(P_2)|=|P_2\cup \tau(P_1)|$.)
We next generalize this result.
For $i\in M=\{1,2,...,m\}$, let $M(-i)=\{1,2,...,i-1,i+1,...,m\}$.
We have the following lemma.

\begin{lemma}
Let $\{i,j\}\subseteq M=\{1,2,...,m\}$.
Let $\tau_i(P_k)$ be the projected points of $P_k$ on $Y=X$ such that these
projected points all have color-$i$, with $k\in M, i\neq k$.
Then,
$$d_B({\cal P}_i,{\cal P}_j)=d^\infty_B(P_i\bigcup\limits_{k\in M(-i)} \tau_i(P_k),P_j\bigcup\limits_{k\in M(-j)} \tau_j(P_k)).$$
\end{lemma}

\begin{proof}
%We only show the first equation, as the other two are similar.
First, notice that, in terms of sizes, we have
$$|P_i\bigcup\limits_{k\in M(-i)} \tau_i(P_k)| = |P_j\bigcup\limits_{k\in M(-j)} \tau_j(P_k)| = \sum\limits_{l=1}^{m}|P_l|.$$

 Following \cite{EH10}, we have
$$d_B({\cal P}_i,{\cal P}_j)=d^\infty_B(P_i\cup \tau_i(P_j),P_j\cup \tau_j(P_i)).$$
Note that 
$$(P_i\bigcup\limits_{k\in M(-i)} \tau_i(P_k))-(P_i\cup \tau_i(P_j))
= (\bigcup\limits_{k\in M(-i)} \tau_i(P_k))-\tau_i(P_j) $$
 and
$$(P_j\bigcup\limits_{k\in M(-j)} \tau_j(P_k))-(P_j\cup \tau_j(P_i))
= (\bigcup\limits_{k\in M(-j)} \tau_j(P_k))-\tau_j(P_i) $$
 are really two sets of identical points
with color-$i$ and color-$j$ on $Y=X$ respectively.
By the definition of infinite multiplicity property of a persistence diagram,
adding these (identical) points on $Y=X$ would not change the bottleneck matching distance between point sets
$P_i\cup \tau_i(P_j)$ and $P_j\cup \tau_j(P_i)$.
Consequently,
$$d^\infty_B(P_i\cup \tau_i(P_j),P_j\cup \tau_j(P_i))
=d^\infty_B(P_i\bigcup\limits_{k\in M(-i)} \tau_i(P_k),P_j\bigcup\limits_{k\in M(-j)} \tau_j(P_k)).$$
Therefore, we have
$$d_B({\cal P}_i,{\cal P}_j)
=d^\infty_B(P_i\bigcup\limits_{k\in M(-i)} \tau_i(P_k),P_j\bigcup\limits_{k\in M(-j)} \tau_j(P_k)).$$
\end{proof}

The implication of the above lemma is that the approximation algorithm in the
previous subsection can be used to compute the approximate center of 
$m$ persistence diagrams. The algorithm can be generalized by
simply projecting each point of color-$i$, say $p\in P_i$, on $Y=X$ to have
$m-1$ projection points with every color $k$, where $k\in M(-i)$.
Then we have $m$ augmented sets $P''_i$, $i=1,...,m$, of distinct colors,
but with the same size $\sum_{l=1..m}|P_l|$. Finally, we simply run Algorithm
1 over $\{P''_1,P''_2,...,P''_m\}$,
with the distance in the $L_\infty$ metric, to have a factor-2 approximation.
We leave out the details for the analysis as at this point all we need is the
triangle inequality of the $L_\infty$ metric.

\begin{theorem}
%\label{thm04}
There is a polynomial time factor-2 approximation for the Center Persistence Diagram problem under the bottleneck distance with $m$ input diagrams
for all the three versions (i.e., `Without Replacement', `With Replacement' and continuous versions).
\end{theorem}

\begin{proof}
The analysis of the approximation factor is identical with Theorem 10.
However, when $m$ is part of the input, each of the augmented point set
$P''_i$, $i=1...,m$, has a size $\sum_{l=1..m}|P_l|=O(mn)$. Therefore,
the running time of the algorithm increases to $O((mn)^{1.5}\log (mn))$,
which is, nonetheless, still polynomial.
\end{proof}

It is interesting to raise the question whether these results still hold if
the $p$-Wasserstein distance is used, which we depict in the next section.
As we will see there, different from under the bottleneck distance,
a lot of questions still remain open.

\section{Center Persistence Diagram under the Wasserstein Distance}

\begin{definition}\textbf{\emph{The Center Persistence Diagram Problem under the $p$-Wasserstein Distance (CPD-W)}}

{\bf Instance}: A set of $m$ persistence diagrams ${\cal P}_1,...,{\cal P}_m$
with the corresponding feature point sets $P_1,...,P_m$ respectively, and a
real value $r$.

{\bf Question}: Is there a persistence diagram ${\cal Q}$ such that $\max_i W_p({\cal Q},{\cal P}_i)\leq r$?
\end{definition}

Note that, similar to CPD-B, we could have three versions depending on ${\cal Q}$: (1) selected
with no replacement from the multiset $\cup_{i=1..m} P_i$, (2)
selected with replacement from the set $\cup_{i=1..m} P_i$, and (3) arbitrarily
selected. We call the first two versions {\em discrete} and the third case
{\em continuous}. Here we will deal with the continuous case as it is still
unknown how to deal with the discrete cases yet.

Given two planar point sets $P$ and $Q$ with size $n$, they naturally form a
complete bipartite graph $\langle P,Q\rangle$. Let $c(P,Q)$ be the {\em weight}
or {\em total cost} of the minimum weight matching in the bipartite graph
$\langle P,Q\rangle$, where the weight or cost of an edge $(p,q)$ is
defines as $c(p,q)=\|p-q\|_2$, for $p\in P,q\in Q$.

\begin{definition}\textbf{\emph{The $m$-BottleneckSum Matching Problem}}

{\bf Instance}: A set of $m$ planar point sets $P_1,...,P_m$ such that
$|P_1|=\cdots=|P_m|=n$, and a real value $r$.

{\bf Question}: Is there a point set $Q$, with $|Q|=n$, such that any
$q\in Q$ is arbitrarily selected and $\max_i c(Q,P_i)\leq r$?
\end{definition}

The `With Replacement' and `With No Replacement' discrete cases can be defined
similarly as in Section 2. But we only focus on this continuous version here.

\subsection{3-BottleneckSum Matching is NP-hard}

In this subsection, we first prove that 3-BottleneckSum Matching is NP-hard.
The crux to modify the proof in Theorem 6 is that the objective function
is to minimize the maximum summation of distances, therefore in Figure 6 around the triple gadget
$T$ the sum of distances from $Q$ to $P_i$'s might be different.
In fact, with the three clusters $(a,b,w)$, $(c,d,u)$ and $(e,f,v)$ alone,
the sum of distances from the corresponding centers, i.e., $b$, $c$ and $e$,
to the points in different colors in these clusters would already be all
different.

\begin{figure}[htbp]
%\psfrag{x1}{$x_1$}
%\psfrag{x2}{$x_2$}
%\psfrag{x3}{$x_3$}
\begin{center}
\includegraphics[scale=0.65]{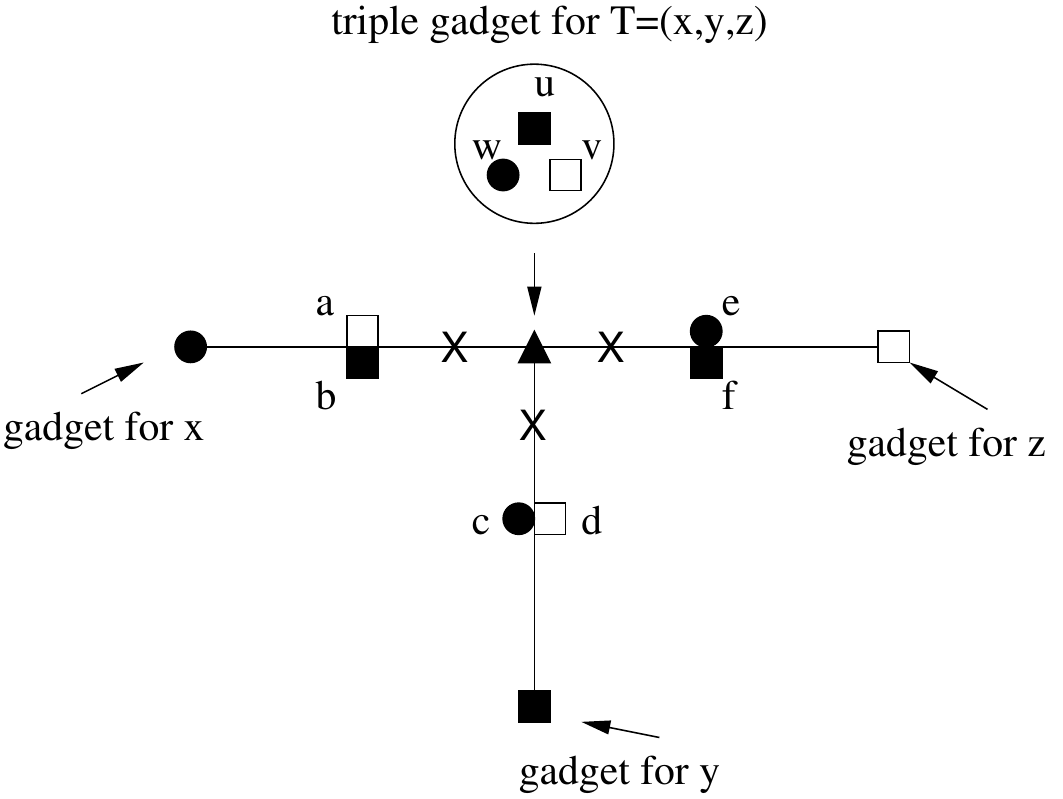}
\end{center}
\label{fig4}
\caption{\bf A skeleton of the reduction from Planar 3DM to 3-BottleneckSum
Matching (and also to the continuous Center Persistence Diagram problem with
$m=3$ diagrams). The triple gadget for $T=\langle x,y,x\rangle$,
represented as $\blacktriangle$, is really putting three element points
of distinct colors on a grid point.}
\end{figure}

The reduction is still from Planar 3DM. The main change is, at each grid edge
with endpoints in color-$i$, we put two points with different colors at the
midpoint of that edge. See Figure 7. Then, the optimal (continuous) center
for a cluster of three points in different colors would be the midpoint of
these three points with different colors, i.e., with minimum radius 1/4. In
Figure 7, the centers marked as `X' are for the clusters $(a,b,w)$, $(c,d,u)$
and $(e,f,v)$ respectively. Note that each of the centers has the same distance
to the three points in distinct colors in the corresponding cluster.
We thus have the following theorem.

\begin{theorem}
The continuous version of 3-BottleneckSum Matching is NP-hard.
\end{theorem}

\begin{proof}
The reduction is similar to that in Theorem 7 and can be done in $O(n^2)$ time.
Let $P_i$ be the set of points used in color-$i$ in the construction, for
$i=1..3$ and with $|P_1|=|P_2|=|P_3|$. Then we could have exactly $|P_1|$
clusters. We claim that $P_1,P_2$ and $P_3$ admit
a point set $Q$ incurring a 3-BottleneckSum matching with a cost of
$(|P_1|-\gamma)/4$ if and only if the Planar 3DM instance has a YES solution. We
leave out the argument details as they are almost identical to those in
Theorem 7.
\end{proof}

We show in the next subsection how Theorem 16 can be extended to the continuous
Center Persistence Diagram problem under the $p$-Wasserstein distance.

\subsection{Continuous Center Persistence Diagram under Wasserstein Distance is NP-hard}

\begin{corollary}
The (continuous) Center Persistence Diagram problem under $p$-Wasserstein distance with $m\geq 3$ input persistence diagrams is NP-hard.
\end{corollary}

\begin{proof}
Our reduction is exactly the same as in Theorem 16.
We then move the constructed points at least $2\hat{D}$ distance away from
$Y=X$ as in Corollary 9, where $\hat{D}$ is the diameter of all the constructed
points on the rectilinear grid.

Let $P_i$ be the set of points used in color-$i$ in the construction, for
$i=1..3$ and with $|P_1|=|P_2|=|P_3|$. As all our $|P_1|$ clusters form either
a horizontal or vertical interval with length 1/2, the $L_\infty$ distance
would be the same as under $L_2$, i.e., each cluster would contribute a value
$(1/4)^p$ toward computing the $p$-Wasserstein distance between ${\cal Q}$ and
${\cal P}_i$ --- using the midpoint of the interval as the corresponding center.
Therefore, we claim that $P_1,P_2$ and $P_3$ admit a center persistence diagram
${\cal Q}$ with a maximum $p$-Wasserstein distance
$\frac{(|P_1|-\gamma)^{1/p}}{4}$ to all $P_i$'s if and only if the Planar 3DM
instance has a YES solution. We again leave out the arguments. This concludes
the proof.
\end{proof}

Note that, for $p<\infty$, the $2-\varepsilon$ inapproximability bound does
not hold anymore as in Corollaries 8 and 9. Moreover, the proof of Corollary 16
does not hold for the discrete versions of CPD-W. Further research is needed
along this line. On the other hand, we comment that the 2-approximation
algorithm still works as the $p$-Wasserstein distance fulfills the triangle
inequality.
 
\section{Concluding Remarks}

In this paper, we study systematically the Center Persistence Diagram
problem under both the bottleneck and $p$-Wasserstein distances. Under
the bottleneck distance, the results are tight as we have a $2-\varepsilon$
inapproximability lower bound and a 2-approximation algorithm (in fact, for
all the three versions). Under the $p$-Wasserstein distance, unfortunately,
we only have the NP-hardness for the continuous version and a 2-approximation,
how to reduce the gap poses an interesting open problem. In fact, a similar
question of obtaining some APX-hardness result was posed in \cite{CKW15}
already, although the (min-sum) objective function there is slightly different.
For the discrete cases under the $p$-Wasserstein distance, it is not even
known whether the problems are NP-hard.
%We comment that in this continuous case, points in ${\cal C}$
%have no relation with any topological feature, so topologically this version
%might not be quite interesting.

\end{document}